\documentclass[english,letterpaper,final]{siamltexmm}

\usepackage{amsthm}
\usepackage{amssymb}
\usepackage{graphicx}
\usepackage{setspace}
\usepackage{amsmath}
\usepackage{thmtools, thm-restate}
\usepackage{xspace}
\usepackage{booktabs}
\usepackage[hyphenbreaks]{breakurl}

\makeatletter

\newcommand \Dotfill {\leavevmode \cleaders \hb@xt@ .33em{\hss .\hss }\hfill \kern \z@ $\,$}

\theoremstyle{plain}

\theoremstyle{definition}

\newtheorem*{defn*}{Definition}

\newcommand{\F}{{\mathcal F}}
\newcommand{\Fk}{{\mathcal F}_\kappa}
\newcommand{\Fkh}{{\mathcal F}_{\hat{K}}}
\newcommand{\unitcube}{{\ensuremath{[0,1]^p}}\xspace}

\newcommand{\E}{{\mathcal E}}
\newcommand{\mmmUncertainty}{\ensuremath{{\mathcal E}_{\hat{K}}}}
\newcommand{\cornerA}{{\mathbf 0}}
\newcommand{\cornerB}{{\mathbf 1}}

\newcommand{\mmfk}{{\hat{f}_{\kappa}}}
\newcommand{\mmfK}{{\hat{f}_K}}
\newcommand{\mmfKhat}{{\hat{f}_{\hat{K}}}}

\newcommand{\Lip}{{\mbox{Lip}}}
\newcommand{\M}{\ensuremath{M_\epsilon}\xspace}

\DeclareMathOperator*{\argmin}{arg\,min}

\DeclareMathOperator*{\dom}{dom}
\usepackage{psfrag}

\makeatother

\title{Mini-Minimax Uncertainty Quantification for Emulators}
\author{Jeffrey C.~Regier\footnotemark[2]
\and Philip B.~Stark\footnotemark[2]}

\begin{document}

\maketitle
\newcommand{\slugmaster}{%
}

\renewcommand{\thefootnote}{\fnsymbol{footnote}}
\footnotetext[2]{Department of Statistics, University of California, Berkeley, CA 94720
(\email{jeff@stat.berkeley.edu}, \email{stark@stat.berkeley.edu})}
\renewcommand{\thefootnote}{\arabic{footnote}}

\pagestyle{myheadings}
\thispagestyle{plain}
\markboth{J.~C. Regier and P.~B. Stark}{Uncertainty Quantification for Emulators}

\begin{abstract}
Consider approximating a ``black box''
function $f$ by an \emph{emulator} $\hat{f}$
based on $n$ noiseless observations of $f$.
Let $w$ be a point in the domain of $f$.
How big might the error $|\hat{f}(w) - f(w)|$ be?
If $f$ could be arbitrarily rough, this error could be arbitrarily
large: we need some constraint on $f$ besides the data.
Suppose $f$ is Lipschitz with known constant.
We find a lower bound on the number of observations required to
ensure that for the best emulator $\hat{f}$ based on the $n$ data,
$|\hat{f}(w) - f(w)| \le \epsilon$.
But in general, we will not know whether $f$ is Lipschitz, much less know its Lipschitz constant.
Assume optimistically that $f$ is Lipschitz-continuous with the smallest
constant consistent with the $n$ data.
We find the maximum (over such regular $f$) of $|\hat{f}(w) - f(w)|$ for the
best possible emulator $\hat{f}$;
we call this the ``mini-minimax uncertainty'' at $w$.
In reality, $f$ might not be Lipschitz or---if it is---it might not attain its Lipschitz constant on the data.
Hence, the mini-minimax uncertainty at $w$ could be much smaller than $|\hat{f}(w) - f(w)|$.
But if the mini-minimax uncertainty is large, then---even if $f$ satisfies the optimistic
regularity assumption---$|\hat{f}(w) - f(w)|$ could be large,
no matter how cleverly we choose $\hat{f}$.
For the Community Atmosphere Model, the maximum (over $w$)
of the mini-minimax uncertainty based on a set of 1154~observations
of $f$ is no smaller than it would be for a single observation
of $f$ at the centroid of the 21-dimensional parameter space.
We also find lower confidence bounds for quantiles of the mini-minimax uncertainty
and its mean over the domain of $f$.
For the Community Atmosphere Model,
these lower confidence bounds are an appreciable fraction of the maximum.
To know that the emulator estimates $f$ accurately would require evidence that $f$ is typically
more regular than it is across the $n$ sample values.
\end{abstract}

\begin{keywords}
emulator, surrogate function, metamodel, minimax, Lipschitz, information-based complexity
\end{keywords}

\begin{AMS}
68Q17, 65D05, 68U20, 62P12 
\end{AMS}

\section{Introduction}

This paper studies the accuracy of emulators, also known as surrogate functions and metamodels.
Emulators are important
tools for approximating functions that have been observed only partially.
Kriging, Multivariate Adaptive Regression
Splines (MARS), Projection Pursuit Regression, Polynomial Chaos Expansions (PC),
Gaussian Process models (GP), and other Bayesian modeling techniques are
common methods for
constructing emulators~\cite{Sacks1989,Ben-Ari2007,Ghanem2008}.
We find error bounds for emulators in general---including the ``best possible'' method---rather
than focusing on any particular emulation method.\footnote{Software that computes the bounds
described in this paper is freely available at \url{https://github.com/jeff-regier/MiniMiniMaxUQ}\,.}

Emulators are frequently used to approximate expensive computer models, which
are often deterministic functions.\footnote{%
   They might not be entirely deterministic; for instance, they could involve
   Monte Carlo simulations.
   Moreover, in distributed parallel computations, numerical results can depend on
   the order in which subproblems happen to complete.
   These cases can be thought of as observing the function with noise.
   We do not address noise here; however, uncertainty
   in the observations makes accurate approximation more difficult.
   Because we focus on lower bounds on the difficulty of approximating the function accurately,
   our results generally remain lower bounds when the observations are not
   only incomplete, but also noisy.
   To extend our methods to include noise would involve
   finding a lower confidence bound on the regularity of the function.
}
Resources limit the number of times the computer model can be run,
though typically an intractable number of inputs is possible---for
instance if any input parameter is a floating point number.
By fitting an emulator to the output of a tractable number of runs for different
inputs, one can approximate the computer model inexpensively;
the issue is the accuracy of that approximation.

Computer models known as \emph{HEB} \cite{Shan2009} may be particularly
difficult to emulate:
They depend on \emph{H}igh-dimensional inputs;
they are \emph{E}xpensive to run; and they are effectively \emph{B}lack
boxes that are not amenable to closed-form, analytic study.
Because such models have high-dimensional inputs, it takes prohibitively many runs
to explore their domains: to attain a given sample density, the number
grows exponentially in the dimension.
Because the models are expensive, performing many runs is impractical or
impossible.
And because the models are black boxes, there are few
(if any) constraints to ensure that the error in extrapolating
from inputs actually tried to inputs not sampled is small.

HEB problems arise often in practice, for instance:
\begin{itemize}
	\item Climate models: \cite{Covey2011} (21--28-dimensional domains;
	1154~simulations; Kriging and MARS)
	\item Automobile crashes: \cite{AspenbergneLonn2012} (15-dimensional domain;
	55~simulations; polynomial response surfaces and artificial
	neural networks).
	\item Chemical reactions: \cite{Holena2011} (30--50-dimensional domain;
	boosted surrogate models) and \cite{Shorter1999}
	(46-dimensional domain; seconds per simulation).
	\item Aircraft design: \cite{Srivastava2004} (25-dimensional domain;
	500 simulations; response surfaces and Kriging), \cite{Koch1999}
	(22-dimensional domain; minutes per simulation;
	response surfaces and Kriging), and \cite{Booker1999} (31-dimensional
	domain; 20~minutes to several days per simulation;
	Kriging).
	\item Electric circuits: \cite{Bates1996} (60-dimensional domain;
	216~simulations; Kriging).
\end{itemize}
How accurately can a function $f$ be emulated from a given set of data?
How many evaluations of $f$ are required to guarantee that $f$ can be emulated
to a given level of accuracy?

Since $f$ is a ``black box,'' we do not know how rough it might be:
extrapolating beyond the data could entail arbitrarily large errors.
We assume that $f$ is regular
and find the resulting uncertainty in emulating $f$.
If the regularity assumption fails,
the uncertainty would be larger.
We measure the regularity of $f$ by
its \emph{absolute condition number} or
\emph{Lipschitz constant} $K$.
Similar results could be derived for other measures of regularity,
but Lipschitz bounds are particularly amenable to analysis.

The observations impose a lower bound $\hat{K}$ on $K$.
Suppose, optimistically, that the true Lipschitz constant
of $f$ is equal to this lower bound.
Then $f$ might be any member of the set $\Fkh$ of functions that agree
with the observations and have Lipschitz constant no greater than $\hat{K}$.
If an emulator is guaranteed to do well no matter which member of $\Fkh$
$f$ happens to be, then the uncertainty of that emulator is low.
On the other hand, if there are elements of $\Fkh$ that an emulator cannot
approximate well, the uncertainty is large.

Consider \emph{all} emulators that can be computed from the observations alone,
without additional knowledge of $f$; this collection includes emulators
constructed using GP, PC, MARS, and all the other methods mentioned above.
Viewed as a function of $w$ in the domain of $f$,
the minimax error among such emulation methods over the set
$\Fkh$ of functions that agree
with the observations and have Lipschitz constant no greater than $\hat{K}$
is the \emph{mini-minimax uncertainty} $\mmmUncertainty(w)$ in the title
of this paper.

The first ``mini'' refers to the regularity condition: since $K$ is not smaller than
$\hat{K}$, $\mmmUncertainty(w)$ is a lower bound on the minimax uncertainty for functions
that are as regular as $f$.
The second ``mini'' refers to emulators: this is the uncertainty for the
best emulator---including all the standard ones.
The ``max'' is over functions that agree with
$f$ at the observations and satisfy the optimistic regularity condition.
That is, $\mmmUncertainty(w)$ is the smallest that the uncertainty at $w$ could be,
for the best emulator, over the set of functions that have the highest degree of
regularity consistent with the observations and that agree with the observations.
The maximum of $\mmmUncertainty(w)$ over $w$ in the domain of $f$ is
an attainable lower bound on the maximum uncertainty of any emulator $\hat{f}$ of $f$.

If $K$ were known, this would be a standard
problem in information-based
complexity \cite{Packel1988,Traub1980,Traub1988}.
We derive bounds on the uncertainty using the lower bound $\hat{K}$
computed from the observed variation of $f$.
Section~\ref{sec:Bounds-on-num-evals}
derives a lower bound on the number of additional observations that
might be necessary to learn $f$.
Section~\ref{sec:Bounds-on-error}
derives two lower bounds on the maximum uncertainty
for approximating $f$ from a fixed set of observations: a purely empirical
bound and a bound expressed as a fraction of the
unknown Lipschitz constant.
The latter yields conditions under which emulating $f$
by a constant function, equal to the value of $f$ at the centroid of its domain,
has smaller maximum uncertainty than any emulator based on the $n$
actual observations.

Section~\ref{sec:applications} applies these bounds to two closed-form functions
(a high-dimensional cone and the \emph{borehole function}~\cite{Surjanovic2014})
and to a black-box function (the Community Atmosphere Model~\cite{Covey2011}).
Section~\ref{sec:extensions} extends the results for the maximum error
to quantiles of the error and the mean of the error over the domain of $f$.
Section~\ref{sec:conclusions} gives our conclusions.

\section{Notation and problem formulation\label{sec:Problem-formulation}}
\begin{table}
	\begin{tabular}{l@{}l@{}}
	    symbol$\qquad\;\;\;\;$ & meaning \\
         \toprule
	    $f$\Dotfill &  unknown function on $\unitcube$ to be emulated \\
	    $\hat{f}$ \Dotfill & an emulator \\
	    $X$ \Dotfill & finite subset of $\unitcube$ where $f$ is observed \\
	    $g|_Y$ \Dotfill & the restriction of the function $g$ to the set $Y \subset \unitcube$ \\
	    $f|_X$ \Dotfill & \emph{the data}: the restriction of $f$ to $X$ \\
	 \midrule
	    $K$ \Dotfill & Lipschitz constant of the function $f$ \\
	    $\hat{K}$ \Dotfill & smallest Lipschitz constant of any function that interpolates the data \\
	    $\F_{\kappa, Y}$ \Dotfill  & all functions that interpolate $f|_Y$ and have Lipschitz
	                              constant no larger \\ & than $\kappa$.  \\
           $\F_{\kappa}$ \Dotfill  & $\F_{\kappa, X}$ \\	  \midrule
	  $e_{\kappa}^{+}(w)$  \Dotfill & maximum value at $w$ among functions in $\F_{\kappa}$ \\
           $e_{\kappa}^{-}(w)$  \Dotfill & minimum value at $w$ among functions in $\F_{\kappa}$  \\
           $\mmfk(w)$ \Dotfill              & mean of  $e_{\kappa}^{+}(w)$ and $e_{\kappa}^{-}(w)$; 
                                                        the minimax emulator
                                                          at the point $w$ over \\ & functions in $\F_{\kappa}$ \\
	  \midrule
	    $\E_{\kappa, Y}(w; \hat{f})$ \Dotfill & \emph{maximum uncertainty of $\hat{f}$ at $w$}: 
	                               uncertainty of $\hat{f}$
	                                    at the point $w$ over \\ & functions
	                                              in $\F_{\kappa, Y}$ \\
	    $\E_{\kappa, Y} (w)$ \Dotfill &    \emph{minimax uncertainty at $w$}:  
	                uncertainty of the best possible emulator at \\ &
	                                               the point $w$ over   functions in $\F_{\kappa, Y}$ \\
	    $\E_{\kappa, Y} (\hat{f})$ \Dotfill &   \emph{maximum  uncertainty of $\hat{f}$}:
	                                                maximum (over $w \in \unitcube$) uncertainty \\& of $\hat{f}$
	                                              over functions in $\F_{\kappa, Y}$  \\
	    $\E_{\kappa, Y}$           \Dotfill &   \emph{minimax  uncertainty}:   maximum (over $w \in \unitcube$)
	                                             uncertainty of \\ & the best possible emulator
	                                              over functions in $\F_{\kappa, Y}$ \\
	    $\E_{\kappa}( \cdots )$ \Dotfill & when $Y = X$, we generally suppress $X$ from the subscript, viz.,
	                          $\E_{\kappa}( w; \hat{f})$, \\ & $\E_{\kappa} ( w)$, $\E_{\kappa} (\hat{f})$, and $\E_{\kappa}$\\
	   \midrule
	   $\M$ \Dotfill & \emph{minimum computational burden}: a lower bound on the number of\\ &
	                           additional observations needed to guarantee that the minimax  \\ &
								uncertainty is no larger than $\epsilon$\\
          \bottomrule
       \end{tabular}
\vspace{10px}
\caption{\protect \label{tab:notation} Summary of key notation}
\end{table}

The function $f$ is a fixed unknown real-valued function on $[0, 1]^p$,
the $p$-dimensional unit cube.
The space of real-valued continuous functions on \unitcube is $\mathcal{C}\unitcube$.
The Roman letters $i$, $j$, $p$, $q$, and $M_\epsilon$ denote integers.
Lowercase Greek letters denote real scalars, with the exception of $\mu$,
which denotes Lebesgue measure.
Uppercase Roman letters such as $X$ and $D$ denote subsets of $[0, 1]^p$;
$X$ is a fixed finite subset of $[0, 1]^p$.
Lowercase Roman letters from the end of the alphabet, such as $v$, $w$, $x$, $y$, and $z$,
denote points in $[0, 1]^p$.
The lowercase Roman letters $e$, $f$, $g$, and $h$ denote real-valued functions
on (subsets of) $[0, 1]^p$.
The domain of a function $g$ is $\dom(g)$.
The restriction of a function $g$ to $D \subset \dom(g)$ is denoted $g|_{D}$.
The observations from which $f$ is to be emulated are $f|_X$; that is, we observe $f$ on
the set $X$.
An emulator $\hat{f}$ is a real-valued function on $[0, 1]^p$.
Let $\| h \|_\infty \equiv \sup_{w \in \dom(h)} | h(w) |$,
the infinity-norm of $h$.
This paper studies how large $|\hat{f}(w) - f(w)|$ and $\|\hat{f} - f\|_\infty$
could be, for the best $\hat{f}$ chosen on the basis of the data---without 
other information about $f$.

Let $d$ be a metric on $\dom(g)$.
The (best) Lipschitz constant for $g$ is
\begin{equation}
\Lip(g) \equiv \sup \left \{
        \frac{g(v)-g(w)}{d(v,w)}:
        {v, w \in \dom(g) \mbox{ and } v \ne w} \right \}.
\end{equation}
If $f \notin \mathcal{C}\unitcube$, then $\Lip(f) \equiv \infty$.
Define
\[
   \F_{\kappa}(g)\equiv\{ (h: \unitcube \rightarrow \Re) :  \Lip(h) \le \kappa \mbox{ and }
   h|_{\dom(g)}=g\}.
\]
Then $\F_{\infty}(f|_X)$ is the space of (possibly discontinuous) functions that fit the $n$ data.
Some of our results involve values of $f$ at points other 
than the points $X$ at which $f$ was observed; 
$Y$ denotes a generic set of points in the domain of $f$.
To simplify notation, we set
\[
     \F_{\kappa, Y}\equiv\F_{\kappa}(f|_{Y}).
\]
When $Y = X$, we generally write $\F_{\kappa}$ in place of $\F_{\kappa, X}$.
\begin{defn*}
The \emph{ uncertainty at $w$ of $\hat{f}: \unitcube \rightarrow \Re$ over
the set of functions $\F_{\kappa, Y}$} is
\[
      \E_{\kappa, Y}( w; \hat{f}) \equiv \sup_{g \in \F_{\kappa, Y}} |\hat{f}(w) - g(w)|
\]
The \emph{minimax uncertainty at $w$ over the set of functions $\F_{\kappa, Y}$} is
\[
      \E_{\kappa, Y} ( w) \equiv \inf_{\hat{f}: \unitcube \rightarrow \Re} \E_{\kappa, Y}( w; \hat{f}).
\]
The \emph{maximum uncertainty of $\hat{f}: \unitcube \rightarrow \Re$ over
the set of functions $\F_{\kappa, Y}$} is
\[
      \E_{\kappa, Y}( \hat{f}) \equiv \sup_{w \in \unitcube} \E_{\kappa, Y}(w; \hat{f})
         = \sup_{g \in \F_{\kappa, Y}} \|\hat{f}-g\|_{\infty}.
\]
The \emph{minimax maximum uncertainty over
the set of functions $\F_{\kappa, Y}$} is
\[
      \E_{\kappa, Y} \equiv \inf_{\hat{f}: \unitcube \rightarrow \Re} \E_{\kappa, Y}( \hat{f}).
\]
\end{defn*}
The emulator $\hat{f}$ approximates $f$ within
$\E_\infty( w; \hat{f})$ at the point $w$
if $f$ is in $\F_{\infty}$, the set of functions
that agree with the observations.
However, $\E_\infty( w; \hat{f})$ is infinite for every $\hat{f}$
unless $w \in X$, even if $f$ is guaranteed to be continuous.\footnote{%
	The set $X$ is not dense in $[0, 1]^p$, so
	for any $c>0$, there exists some function $g\in\F_{\infty}(f|_{X})$
	satisfying $\|f-g\|_{\infty}>c$.
}
To guarantee that the uncertainty is finite requires stronger
regularity than mere continuity.

Let $K \equiv \Lip(f)$ and $\hat{K} \equiv \Lip(f|_X)$.
Because $X\subset\unitcube$, $\hat{K}\le K$, as illustrated in figure~\ref{fig:k_vs_khat}.
(There and in subsequent figures, $p=1$ and the bold
black dots represent $f|_X$, the observations of $f$ at $x \in X$.)

\begin{figure}
\centering
\psfrag{f}{$f$}
\includegraphics{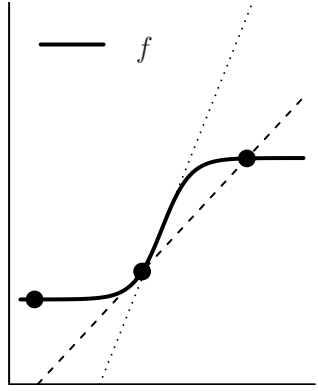}
\caption{\label{fig:k_vs_khat}
    Illustration of the difference between the true Lipschitz constant $K$ and the
    empirical lower bound $\hat{K}$ for $K$.
    The dotted line is tangent to $f$ where $f$
	attains its Lipschitz constant: it has slope $K$.
	The dashed line is the steepest line that intersects any pair of observations:
	it has slope $\hat{K} \le K$.}
\end{figure}

Define
\[
   e_{\kappa}^{+}(w)\equiv
   \min_{x\in X}\left[f(x)+\kappa d(x,w)\right]
\]
and
\[
   e_{\kappa}^{-}(w)\equiv
   \max_{x\in X}\left[f(x)-\kappa d(x,w)\right].
\]
The mean of the two is
\[
     \mmfk(w) \equiv \mmfk(w; X, \kappa) \equiv \frac{e_{\kappa}^{-}(w) + e_{\kappa}^{+}(w)}{2}.
\]
Figures~\ref{fig:sin-panels} and \ref{fig:tilt-panels} illustrate
these definitions.
The proof of Proposition~\ref{prop:supequiv} shows that
the function $\mmfk(w)$ is the minimax emulator for pointwise error over the class 
$\Fk$ of functions that agree with the data and have Lipschitz constant
no greater than $\kappa$.
The minimax emulator $\mmfk(w)$ interpolates (rather than smooths) the data.

\begin{figure}
\centering

\psfrag{ep}{$e_{\kappa}^+$}
\psfrag{fs}{$\hat{f}_{\kappa}$}
\psfrag{em}{$e_{\kappa}^-$}
\psfrag{f}{$f$}

\includegraphics{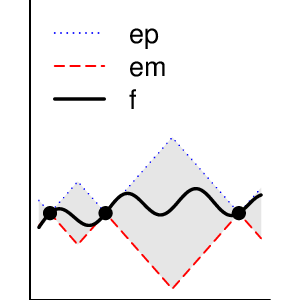}\includegraphics{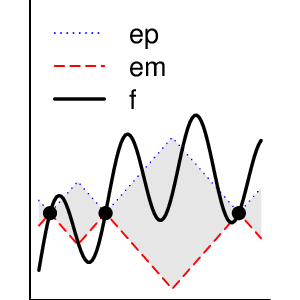}

\caption{\label{fig:sin-panels}
    Illustration of the upper and lower envelope functions $e_{\kappa}^{-}$ and $e_{\kappa}^{+}$.
    In the left panel, $\kappa=K$; in the right, $\kappa<K$. 
    If $\kappa\ge K$, then $e_{\kappa}^{-}\le f\le e_{\kappa}^{+}$, and, equivalently, $f \in \F_{\kappa}$.
	\vspace{10px}
}
\end{figure}

\begin{figure}
\centering

\psfrag{ep}{$e_{\kappa}^+$}
\psfrag{em}{$e_{\kappa}^-$}
\psfrag{fs}{$\hat{f}_{\kappa}$}
\psfrag{f}{$f$}

\includegraphics{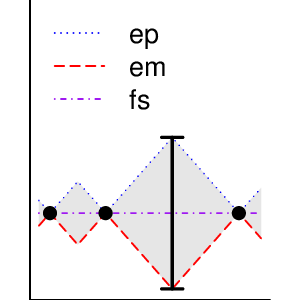}\includegraphics{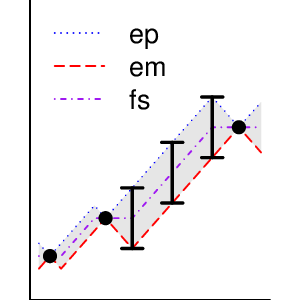}\includegraphics{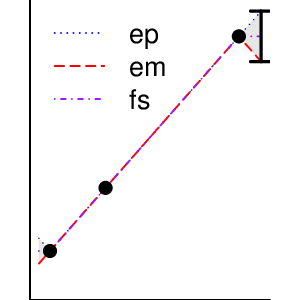}

\caption{\label{fig:tilt-panels}
Illustration of how the pointwise uncertainty depends on the observed variation
of $f$: the uncertainty is smaller where the data require $f$ to vary rapidly.
The vertical distance between the blue and
red curves is twice the uncertainty at the corresponding abscissa.
The black error bars are at some points where the uncertainty is largest.
The succession of panels shows that as the slope between observations
approaches $\kappa$, $\E_{\kappa}(w)$ approaches $0$ for points $w$
between observations, and the maximum uncertainty decreases.}

\end{figure}

\section{Bounds on the number of observations needed to approximate $f$ well}
\label{sec:Bounds-on-num-evals}

In this section we construct a function $\bar{f}$ that agrees with the data $f|_X$,
has Lipschitz constant $\hat{K}$ (the smallest Lipschitz constant consistent with the data),
and yet would require a large number $\M$ of additional observations $f|_Y$ to estimate $f$ within
$\epsilon$ on $\unitcube$.\footnote{%
   We do not discuss the choice of $\epsilon > 0$ in detail:
   scientific context should inform the choice.
   In examples below, we set $\epsilon$ to be an absolute tolerance,
    a fraction of $\hat{K}$, and a fraction of $K$.
  One might also consider relating $\epsilon$ to the
  ``typical value'' of $f$ (e.g., the mean of $f$ or of $f|_X$).
}
The function $\bar{f}$ is not intended to be an emulator---it is a technical device.
Since $f$ could in fact be $\bar{f}$, this gives a lower bound on the number of additional
observations that might be required to estimate $f$ well, even if $f$ is no rougher than the
original data $f|_X$ reveal it to be.

Let $B(x,\delta)$ denote the open ball in $\mathbb{R}^{p}$ centered
at $x$ with radius $\delta$.
Since $f$ has Lipschitz constant $K$, $f(y)$ is guaranteed to be within $\epsilon$ of $f(x)$
if $y \in B(x, \epsilon/K)$.
But depending on $f$ and $X$, it can happen that $\mmfKhat$ is guaranteed to be within
$\pm \epsilon$ of every $g \in \F_{K}$ for parts of the domain not contained in
$\cup_{x \in X} B(x, \epsilon/K)$.
To see this, consider $p=1$, $f(x)=x$, and let $X$ be the two-element set $\{0,1\}$.
Then $K=\hat{K}=1$.
In this case, the observations $f|_{X}$ determine $f$ exactly:
the only function in $\F_{K}$ is $f$.
In this example, for a function $g$ to agree with the
observations requires it to attain the Lipschitz constant $K$ everywhere.
A function cannot agree with the observations and ``run away'' from $f$ very far.

More generally, if $f$ varies on $X$, then for a function $g$ to agree with $f$
at the observations, $g$ must vary too.
That required variation ``spends'' some of $g$'s Lipschitz constant,
preventing $g$ from running as far away from $f$ as it could if $f_X$ were constant.
We now quantify this intuition to construct a function $\bar{f}$ that requires many
additional observations to estimate well.
The function $\bar{f}$ is constant 
``as much as possible'' subject to the constraint that it interpolates
the data and has Lipschitz constant $\hat{K}$.
Since estimating $\bar{f}$ where it is constant is hard 
(as illustrated in figure~\ref{fig:tilt-panels}), 
the size of the set where
$\bar{f}$ could be constant gives a lower bound on the  
number of additional observations that might be required.

Define $\bar{\gamma}\equiv\argmin_{\gamma\in\mathbb{R}}\sum_{x\in X}\left|f(x)-\gamma\right|^{p}$.
Computing $\bar{\gamma}$ is straightforward because the objective
function is univariate and convex.\footnote{%
   Alternatively,  we could set
   $\bar{\gamma}\equiv\frac{1}{\#X}\sum_{x\in X}f(x)$, where $\#X$
   is the size of $X$.
   The resulting lower bound may not be as tight.
}
Let $X^{+}\equiv\{x\in X:f(x)\ge\bar{\gamma}\}$ and let $X^{-}\equiv\{x\in X:f(x)<\bar{\gamma}\}$.
Let
\[
          Q_{+}\equiv\bigcup_{x\in X^{+}}\left\{ B\left(x,\frac{f(x)-\bar{\gamma}}{\hat{K}}\right)
              \bigcap\unitcube\right\}
\]
and
\[
           Q_{-}\equiv\bigcup_{x\in X^{-}}\left\{ B\left(x,\frac{\bar{\gamma}-f(x)}{\hat{K}}\right)
                \bigcap\unitcube\right\} .
\]
Then $Q_{+}\cap Q_{-}=\emptyset$.%
\footnote{%
	Fix $x^{+}\in X^{+}$ and $x^{-}\in X^{-}$.
    Then $\left|f(x^{+})-f(x^{-})\right|/
		{d(x^{+},x^{-})}\le\hat{K}$.
	Equivalently, $d(x^{+},x^{-})\ge
		\left|f(x^{+})-f(x^{-})\right|/{\hat{K}}$.
	Let $B^{+}=B\left(x^{+},
		\left[f(x)-\bar{\gamma}\right]\hat{K}\right)$
	and $B^{-}=B\left(x^-,
		\left[\bar{\gamma}-f(x)\right]/\hat{K}\right)$.
	Let $a$ be the sum of the radii of $B^{+}$ and $B^{-}$.
	Then $a=\left(f(x^{+})-\bar{\gamma}\right)/\hat{K}+
		\left({\bar{\gamma}-f(x^{-})}\right)/\hat{K}
		=\left({f(x^{+})-f(x^{-})}\right)/\hat{K}$,
	and $a\le d(x^{+},x^{-})$. Therefore, $B^{+}\cap B^{-}=\emptyset$.
	Because our selection of $x^{+}\in X^{+}$ and $x^{-}\in X^{-}$ was
	arbitrary, $Q^{+}\cap Q^{-}=\emptyset$.%
}

Define
\begin{align*}
\bar{f}:\unitcube & \to\mathbb{R}\\
w & \mapsto
  \left \{
      \begin{array}{ll}
             e_{\hat{K}}^{-}(w),  & w\in Q_{+}\\
             e_{\hat{K}}^{+}(w), & w\in Q_{-}\\
             \bar{\gamma}, & \mbox{otherwise}.
      \end{array}
   \right .
\end{align*}
Figure \ref{fig:fbar} illustrates this definition. If we know $f|_{X}$,
we know $\bar{f}$.
By construction, $\bar{f} \in \F_{\hat{K}} \subset \F_{K}$.

\begin{figure}
\centering

\psfrag{ep}{$e_{\hat{K}}^+$}
\psfrag{em}{$e_{\hat{K}}^-$}
\psfrag{fb}{$\bar{f}$}
\psfrag{yb}{$\bar{\gamma}$}

\includegraphics{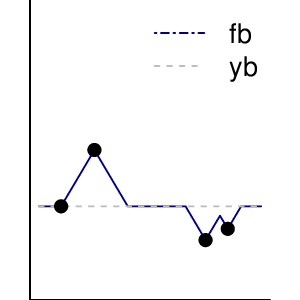}\includegraphics{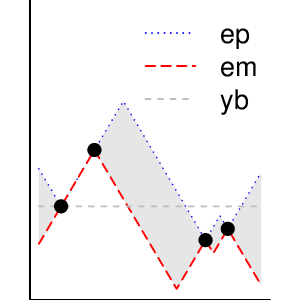}

	\caption{\label{fig:fbar}A function that agrees with the data, has Lipschitz constant $\hat{K}$,
	and is hard to estimate because it is often constant.
	The function $\bar{f}$ (shown in the left panel)
	is comprised of segments of $e_{\hat{K}}^{+}$, $e_{\hat{K}}^{-}$
	and the constant function $\bar{\gamma}$ (all shown in the right panel).
	It is constant over roughly half of the domain.
	No function between $e_{\hat{K}}^{-}$ and $e_{\hat{K}}^{+}$ (inclusive) is constant
	over a larger fraction of the domain.}

\end{figure}

Let $\bar{Q}\equiv\unitcube\setminus(Q_{+}\cup Q_{-})$. Let $\mu$
be Lebesgue measure.
By the union bound, because $\mu(\unitcube)=1$,
\[
      \mu(\bar{Q})\ge1-\sum_{x\in X}\mu\left(B\left(x,  |f(x)-\bar{\gamma}|/\hat{K} \right)\right).
\]
Let $C_{2}\equiv\frac{\pi^{p/2}}{\Gamma(p/2+1)}$ and $C_{\infty}\equiv2^{p}$,
where $\Gamma$ is the gamma function.
Then, for $q\in\{2,\infty\}$,
\[
       \mu(\bar{Q})\ge1-C_{q}\sum_{x\in X}\left(|f(x)-\bar{\gamma}|/\hat{K} \right)^{p}.
\]
If there is some $x \in X$ for which for all $g \in \F_{\hat{K}, \{x\}}$,
$|g(y) - f(x)| < \epsilon$ for all $y \in A\subset\bar{Q}$, then $\mu(A)\le$$\mu(B(0, \epsilon/\hat{K}))$.
Hence, because $\bar{f}\in\F_{K}$,
\begin{eqnarray}
   \M & \ge & \left\lceil \frac{\mu(\bar{Q})}{\mu(B(0,\epsilon/\hat{K}))}\right\rceil \nonumber \\
 & \ge & \left\lceil \epsilon^{-p}\left[\frac{\hat{K}^{p}}{C_{q}}-\sum_{x\in X}|f(x)-\bar{\gamma}|^{p}\right]\right\rceil .\label{eq:M_lower_bound}
\end{eqnarray}
Section~\ref{sec:applications} shows that this lower bound, the \emph{minimum computational burden}, 
can be extremely large for even modest problem dimensions $p$.

\section{Bounds on the maximum uncertainty for a fixed experimental design\label{sec:Bounds-on-error}}
The previous section gave lower bounds on the number of additional observations of $f$
required to attain a desired maximum uncertainty $\epsilon$.
This section gives two lower bounds on the maximum uncertainty
$\E_{K} (\hat{f})$
for a fixed experimental design $X$:
an absolute bound and a bound expressed as a fraction of $K$.
The bound as a fraction of $K$ can yield a strong negative result: when
a statistic---calculable from the observations---exceeds a calculable
threshold, the maximum uncertainty is not less than the maximum uncertainty
of the best emulator based on a single observation at the centroid of the domain.
If the goal is to minimize the maximum uncertainty, we could have just approximated
$f$ as constant and saved $\#X - 1$ observations.

\subsection{Lower bounds}
%
Consider the set $\F_\kappa$ of functions $g$ that agree with the observations $f|_X$ and
have Lipschitz constant no larger than $\kappa$.
Consider all possible emulators $\hat{f}$.
Proposition~\ref{prop:supequiv} states that
the smallest (across emulators $\hat{f}$) maximum (across functions $g$)
error at the point $w \in \unitcube$ is $[e_{\kappa}^+(w) - e_{\kappa}^-(w)]/2$, and
the emulator $\mmfk(w)$ attains this bound at every $w$.

\begin{restatable}{prop}{minimaxEquivProp}\label{prop:supequiv}
  If $\kappa \ge \hat{K}$, then
\[
	\E_{\kappa}(w) = \E_{\kappa}(w; \mmfk) = \frac{e_{\kappa}^+(w) - e_{\kappa}^-(w)}{2}.
\]
\end{restatable}
\noindent
Proofs, including the proof of Proposition~\ref{prop:supequiv}, are in appendix~\ref{sec:proofs}.

\begin{restatable}{cor}{empiricallbcor} \label{cor:empiricallb}
For any emulator $\hat{f}$,
\begin{equation}
   \E_K(w, \hat{f}) \ge \E_K(w)  \ge  \E_{\hat{K}}(w; \mmfKhat).
\end{equation}
\end{restatable}
\noindent
Corollary~\ref{cor:empiricallb} follows from proposition~\ref{prop:supequiv}
and the fact that, since $\hat{K}\le K$,
\[
   \F_{\hat{K}} \subset \F_{K}.
\]
Corollary~\ref{cor:empiricallb} is one of our principal results:
$\E_{\hat{K}}$, a statistic
calculable solely from the observations $f|_X$, is a lower bound on the maximum uncertainty
for \emph{any} emulator $\hat{f}$ based on the observations $f|_X$.
Theorem~\ref{thm:main} gives a stronger lower bound in terms of the unknown value of $K$.
\begin{restatable}{thm}{thmmain}
\label{thm:main}
	For any $\lambda \in \Re^+$, \textup{if $\mmmUncertainty \ge \lambda \hat{K}$,
	then $ \E_K( \hat{f})\ge  \lambda  K$.}
\end{restatable}

\subsection{\label{sub:one-observation}Maximum uncertainty for an emulator based on one observation}
In this section we work in $\ell_\infty$: $d(v,w)=\|v-w\|_{\infty}$.
This simplifies the calculations and gives a particularly strong result.

Let $z\equiv\left(1/2,\ldots,1/2\right)$, the centroid
of \unitcube, and let $Z \equiv \{z\}$.
Let $\hat{g}\in\F_{\infty, Z}$
be the constant function $\hat{g}(w)\equiv f(z)$, $\forall w \in \unitcube$.
The $\ell_\infty$ distance from $z$ to any point on the boundary
of \unitcube is $1/2$, so
\[
   \E_{K,Z}( \hat{g})=\frac{K}{2}.
\]
That is, the maximum uncertainty of the emulator that is constant throughout 
$\unitcube$ and equal to the value of $f$ at the centroid of the cube is $K/2$.
Let $W\subset\unitcube$ be finite and $c\in\mathbb{R}$.
Suppose $f$ is constant on the set $W$ and that $W$ contains fewer than $2^{p}$ points.
Let $\hat{h} \in \F_{\infty, W}$.
By examining the corners of the domain, it follows that 
\[
    \E_{K, W}(\hat{h}) \ge \frac{K}{2}.
\]
Making $2^p$ observations of $f$ is intractable for the Community Atmosphere Model 
and for many other applications.
If $f$ is nearly constant, the situation may still be hopeless.

How do we know whether $f|_{X}$ is too close to constant to benefit
from observing it more than once, but fewer than $2^{p}$ times?
\begin{restatable}{cor}{mainthmcor} \label{cor:mainthm}
If
$\E_{\hat{K}} \ge \hat{K}/2$, then
\[
      \E_K( \hat{f}) \ge \frac{K}{2} \ge \E_{K,Z}(\hat{g}).
\]
\end{restatable}
\noindent
That is, if $\E_{\hat{K}}\ge \hat{K}/2$,
no emulator based on observing $f|_X$ has smaller maximum uncertainty than
the constant emulator based on a single observation---$f$ is too nearly constant.
Corollary~\ref{cor:mainthm} follows directly from theorem~\ref{thm:main},
taking $\lambda  = \hat{K}/2$.

\section{Applications\label{sec:applications}}
This section presents three examples of increasing complexity: 
two in which $f$ is known analytically, and one in which
$f$ is \emph{HEB} arising from a numerical model of climate.
In this section the distance metric is $d(v, w) = \| v - w \|_\infty$, except where noted.

\subsection{High-dimensional $\ell_\infty$ cone}
Consider a emulating a function defined on the 21-dimensional hypercube $[0,1]^{21}$;
$z \equiv \left(0.5,\ldots,0.5\right)$ denotes the center of that hypercube.
Suppose
\[
	f(x) \equiv \|x-z\|_{\infty}.
\]
We observe $f$ at $z$ and, for $i=1, \ldots 21$, at both
points satisfying $x_{i}\in\{0,1\}$ and $x_{j}=0.5$ for $j\ne i$.
(This is a ``one-at-a-time'' sampling design, where one component at a time is shifted from
a typical value to a more extreme value.)
These 43~points constitute $X$.
Then
\[
	\hat{K} = K = 1.
\]
Because every point $w\in[0,1]^{21}$ is within $0.5$ of $x\in X$
satisfying $f(x)=0.5$,
\[
	e_{\hat{K}}^{-} \ge 0.
\]
Because every point $w\in[0,1]^{21}$ is within $0.5$ of $z$, and
$f(z)=0$,
\[
	e_{\hat{K}}^{+} \le 0.5.
\]
Hence, by corollary~\ref{cor:empiricallb},
\[
	\mathcal{E}_{\hat{K}} \le 0.25.
\]
Had we only observed $f$ at $z$ but fixed $\hat{K}=1$ (or observed
$f$ at another point in addition to $z$ and computed $\hat{K}$ from those two points),
\[
	\mathcal{E}_{\hat{K}} = 0.5.
\]
In this example, despite the high dimension of $\dom(f)$, emulating $f$ using
a modest number of observations (43)
has smaller maximum uncertainty than emulating $f$ using just a
single observation of $f$ at $z$: a small number of observations
may constrain a high-dimensional function globally.
High-dimensional problems with small numbers of data do not necessarily have large uncertainties,
as  ``the curse of dimensionality'' would suggest.
The dimension matters, but so does $f$ itself.

To connect our results to  a common emulation method, we fit a Gaussian
process to $f|_X$ by maximum likelihood using the R package {\tt mlegp}~\cite{Dancik2013}.
For 100,000 points selected uniformly at random from $\unitcube$, the
mean error is $0.02$, $2$\% of $K$.
The maximum error at these 100,000 points is $0.23$,
but the error at $(0.6,\ldots, 0.6)$---which is not in the sample---is $0.38$.\footnote{%
   This point was found by searching the ray $c(1,\ldots, 1)$; there might be points with
   even larger errors.
}
Because the error of $\hat{f}_{\hat{K}}$ is no greater than $\mathcal{E}_{\hat{K}} = 0.25$,
for this $f$, the minimax emulator $\hat{f}_{\hat{K}}$ outperforms this 
Gaussian process emulator both in minimax uncertainty and in actual maximum error.

\subsection{Borehole function}

The commonly used test function
\[
	f_0(H_u, H_\ell, T_u, T_\ell, r, r_w, L, K_w)
	       \equiv \frac{2\pi T_{u}\left(H_{u}-H_{l}\right)}
		{\log\left(r/r_{w}\right)\left(1+\frac{2LT_{u}}{\log\left(r/r_{w}\right)r_{w}^{2}K_{w}}+
		\frac{T_{u}}{T_{\ell}}\right)}
\]
models water flow through a borehole~\cite{Surjanovic2014}.
Its input variables are described in table~\ref{tab:borehole_domain}, which also lists the ranges of
those variables.
The output is water flow rate in cubic meters per year.
We rescale $f_0$ so that its inputs range over the 8-dimensional unit hypercube $[0,1]^8$;
the resulting function is denoted $f$.

\begin{table}
\centering
	\caption{ \protect{\label{tab:borehole_domain}}
	Borehole function domain}
	\begin{tabular}{l|l|l}
		variable & range & description \\
		\hline
		$H_u$ & $[990, 1110]$ & potentiometric head of upper aquifer (m) \\
		$H_\ell$ & $[700, 820]$ & potentiometric head of lower aquifer (m) \\
		$T_u$ & $[63070, 115600]$ & transmissivity of upper aquifer ($\text{m}^2/$yr) \\
		$T_\ell$ & $[63.1, 116]$ & transmissivity of lower aquifer ($\text{m}^2/$yr) \\
		$r$ & $[100, 50000]$ & radius of influence (m) \\
		$r_w$ & $[0.05, 0.15]$ & radius of borehole (m) \\
		$L$ & $[1120, 1680]$ & length of borehole (m) \\
		$K_w$ & $[9855, 12045]$ & hydraulic conductivity of borehole (m/yr)
	\end{tabular}

\end{table}
Reasoning about the functional form of $f$ (appendix \ref{sec:borehole_K})
shows that
\[
	944 \le K \le 1200.
\]
Of course, if $f$ really were a black box, such reasoning would be impossible.
We estimated $\hat{K}$ from 1000 sample points selected in two different ways:
\begin{enumerate}
   \item Select 1000 points by Latin hypercube sampling.
           This yields $\hat{K} = 367$.
  \item Select 100 points by Latin hypercube sampling.
           For each of these points,  draw an additional 9~points a small
           distance ($10^{-5}$) from it in each coordinate, in a random
           direction.
           This yields $\hat{K} = 576$.
\end{enumerate}
We fix $\hat{K} = 576$ for the remainder of this example; note that this is roughly half the true
value of $K$.

Now let $X$ contain the following 273~points:
all $2^{8}=256$ corners of $\left[0,1\right]^{8}$,
the center of the domain $(0.5, \ldots, 0.5)$, and, for
$i=1, \ldots, 8$, each of the two points satisfying $x_{i}\in\{0,1\}$ and $x_{j}=0.5$
for $j\ne i$.
(The empirical Lipschitz constant of $f$ on this set is less than $576$.)
By branch-and-bound we find
\[
	\mathcal{E}_{576} < 207
\]
which is less than $576/2$.
Hence, by corollary~\ref{cor:empiricallb}, the best emulator $\hat{f}_{576}$
based on $f|_X$ has lower maximum uncertainty than the best
emulator based on $f|_{\{z\}}$ alone.

Holding $X$ fixed, we now lower-bound $\M$,
the minimum computational burden (section~\ref{sec:Bounds-on-num-evals}).
Convex programming finds $\bar{\gamma}=134.7$.
The union bound implies that the proportion of the domain where
$f$ could be constant is $\mu\left(\bar{Q}\right) \ge 0.76$.
Then for $\epsilon = 100$ (about 20\% of $\hat{K}$ or 10\% of $K$),
$M_{100} \ge 3598$~additional observations might be needed.
But for $\epsilon = 10$, $M_{10} \ge 3.59\times10^{11}$
additional observations might be required.

For comparison, we emulate $f$ by a Gaussian process, again estimating the parameters using the
R package {\tt mlegp}~\cite{Dancik2013} from the same set $X$ of 273~points.
For 100,000 points selected at random uniformly from $[0, 1]^8$, the mean error is
$37.3$, approximately $3$\% of $K$.
The maximum error at these points is $207.9$, approximately $20$\% of $K$.

\subsection{\label{sub:climate_app_1}Climate modeling}

The Uncertainty Quantification Initiative at Lawrence
Livermore National Laboratory\footnote{%
	This dataset was provided by the
	Institutional Science and Technology Office at
	Lawrence Livermore National Laboratory
	under the Uncertainty Quantification Strategic Initiative
	Laboratory-Directed Research and Development Project 10-SI-013.
}
provided results from $1154$ climate simulations using the
Community Atmosphere Model (CAM) with $p=21$ parameters.
Each parameter was scaled so that the interval $[0,1]$
contained all values considered physically reasonable.
The output of interest was a scalar, the simulated global
average upwelling longwave flux (FLUT) averaged over the third through twelfth
years of the simulation (a $10$-year average after a 2-year burn-in).
Each such average is deterministic: repeating a run with the same input parameters
should produce the same output.
The simulator amounts to a function $f$ that maps
$\unitcube \rightarrow \mathbb{R}$.
Running the simulator was computationally expensive; each run took several
days on a supercomputer.
The Lawrence Livermore National Laboratory team used several approaches
to choose the points $X \subset \unitcube$
at which to run simulations, including
Latin hypercube, one-at-a-time, and random-walk multiple-one-at-a-time~\cite{Covey2011}.
The 1154~simulations include all points selected by any of those approaches.

For these observations, we find $\bar{\gamma}=232.77$, 
$\hat{K}=14.20$ for $q=2$, and $\hat{K}=34.68$ for $q = \infty$.

\subsubsection{Computational burden}
By (\ref{eq:M_lower_bound}),
\[
\M \ge\left\lceil \epsilon^{-21}
     \left[\frac{1.57\times10^{24}}{0.014}-6.81\times10^{24}\right]\right\rceil >\epsilon^{-21}\times10^{25}
\]
for $q = 2$.
For example, if $\epsilon$ is $1\%$ of $\hat{K}$, then $\M \ge10^{43}$.
Even if $\epsilon$ is $50\%$ of $\hat{K}$, $\M >10^{8}$.
For $q=\infty$, 
\[
\M \ge \left\lceil \epsilon^{-21}
          \left[\frac{2.19\times10^{32}}{2^{21}}-6.81\times10^{24}\right]\right\rceil >
          \epsilon^{-21}\times10^{25}.
\]
These lower bounds on the minimum computational burden are extreme
for a wide range of values of $\epsilon$: there are functions that fit the 1154~observations
and are as regular as the observations allow, but that cannot be approximated
with useful uncertainty from any tractable number of observations.
The function $\bar{f}$, which is simple to construct, attains these lower
bounds on minimum computational burden.
Note the contrast with the cone example, which was also 21-dimensional: 
the dimension of $\dom(f)$ does not by itself determine how hard it is to emulate $f$ accurately.

\subsubsection{\label{sub:Application2} Uncertainty}
Is the maximum uncertainty of the best emulator based on
observing $f$ at the $1154$ points in $X$
lower than the maximum  uncertainty of the
constant emulator based on one observation of $f$ at the centroid of $\unitcube$?
We cannot simply compute these two maximum  uncertainties,
because $K$ is unknown.
But corollary~\ref{cor:mainthm} applies if we can determine
whether $\E_{\hat{K}}\ge \hat{K} / 2$.
Unfortunately, determining $\E_{\hat{K}}$ is difficult.
In $\ell_\infty$, if $f|_X$ is constant,
finding $\E_{\hat{K}}$ amounts to finding a maximal empty hypercube,
a problem recently shown to be NP-hard in $p$~\cite{Backer2010}.
It is generally no easier if $f$ varies on $X$.
Fortunately, it suffices to bound $\E_{\hat{K}}$.
By working in $\ell_\infty$, we can bound $\E_{\hat{K}}$
above and below by considering just the corners of $\unitcube$;
we take $d(v,w)=\|v-w\|_{\infty}$ throughout this section.

\begin{restatable}{prop}{climateubprop}
\label{prop:climateub}
Let $ \cornerA \equiv(0, \ldots, 0)$, $ \cornerB \equiv(1, \ldots, 1)$,
and $\tilde{d}(v ) \equiv \max \left(d(v , \cornerA ),d(v , \cornerB )\right)$.
Then
\[
\E_{\hat{K}} \le
	\frac{1}{2}\left\{
		\min_{x\in X}\left[f(x)+\hat{K}\tilde{d}(x)\right]
			-\max_{x\in X}\left[f(x)-\hat{K}\tilde{d}(x)\right]\right\}.
\]
\end{restatable}

\noindent
Using this proposition, we calculate $\E_{\hat{K}} \le 20.95$ for the
CAM dataset.
On the other hand, the maximum over all $\unitcube$ is at least as large as 
the maximum over the corners of $\unitcube$:
\begin{eqnarray*}
\E_{\hat{K}}\ge\max\left\{ \E_{\hat{K}}(w):\forall w\in\{0,1\}^{p}\right\}.
\end{eqnarray*}
Perhaps surprisingly, this lower bound is essentially sharp for the CAM dataset.
The domain \unitcube contains $2^{p}$ corners $\{r_{i}\}_{i= 1}^{2^p}$.
Divide $\unitcube$ into $2^p$
hypercubes $\{R_i\}_{i=1}^{2^p}$ with edge-length $1/2$, disjoint interiors,
each containing a different corner of $\unitcube$
(e.g., one such hypercube is $[0, 1/2]^p$).
Then the $R_{i}$ are disjoint $\ell_{\infty}$-balls
of radius $1/4$.
Because $X$ contains only $1154$ points,
the vast majority of $\{R_{i}\}_{i=1}^{2^p}$ do not contain any element of $X$.
Because $\E_{\hat{K}}(w)$ tends to increase with distance from
points in $X$, these unoccupied hypercubes are good regions to
look for points with large values of $\E_{\hat{K}}(w)$.
Within an unoccupied hypercube $R_{i}$, no point is farther in $\ell_{\infty}$
from any point in $X$ than the corner $r_{i}$.
So, the corners $\{r_{i}\}_{1}^{2^p}$
are good places to observe $\E_{\hat{K}}(w)$ to
find a tight lower bound on $\E_{\hat{K}}$.

For the CAM dataset, one corner $r_j$ attains
$\E_{\hat{K}}(r_j)=20.95$.
Since this is also the numerical upper bound, $\E_{\hat{K}}=20.95$.


Because $\E_{\hat{K}} \ge \hat{K}/2 = 17.34$,
theorem~\ref{thm:main} says that $ \E_K( \hat{f})\ge K/2$ for
any emulator $\hat{f}$.
In other words, by the discussion in
section \ref{sub:one-observation}, our maximum  uncertainty would have
been no greater had we just observed $f$ once, at $z$, and predicted
$\hat{f}(w)=f(z)$ for all $w \in \unitcube$.

In some sense, this result is not surprising: if we had fixed $\hat{K}$
but replaced $f$ with a constant function, and $\#X<2^{p}$, then
$\E_{\hat{K}}\ge \hat{K}/2$, with equality holding
if and only if $z \in X$. By repeating the bounding procedures from
the previous two sections with $\hat{K}/2=17.34$ fixed but
$f$ replaced with constant function $c$, we find $\E_{c,X,\hat{K}}=26.95$.
The increase in maximum  uncertainty from 20.95 to 26.95 that results from
replacing $f$ with a constant shows that the observed variation in
$f$ reduces the maximum  uncertainty considerably---although
the maximum  uncertainty remains quite large.

To connect these theoretical results to common emulation methods, 
we fit a Gaussian process model~\cite{Dancik2013} and 
Multivariate Adaptive Regression Splines (MARS)~\cite{Hastie2013} 
to the 110~CAM observations from a Latin hypercube design, leaving 1043 observations
for testing.
On the test set, the mean error of the Gaussian process 
model is $1.03$ ($3$\% of $\hat{K}$) and the maximum error is $6.73$ ($20$\% of $\hat{K}$). 
For MARS, the mean error on the test set is $1.59$ and maximum error is $6.21$. 
Since the 1043~test points are all distant from many corners of $\unitcube$,
the error of these methods over $\unitcube$ might be far larger; 
it would take many more evaluations of $f$ to tell.
Absent such data, there is no evidence that those methods have maximum error less than
$\E_{\hat{K}} = 20.95$.

\section{Extensions} \label{sec:extensions}
\subsection{Distribution of the uncertainty}
By drawing independent points 
$W \sim \mbox{Uniform}(\unitcube)$ and 
evaluating $\E_{\hat{K}}(W)$,
we construct lower confidence bounds for quantiles of the uncertainty and
the mean uncertainty over $\unitcube$.
Table~\ref{tab:integrals} shows the results for the CAM simulations based on 10,000
random samples from $\unitcube$.
Even the lower quartiles are a large fraction of $\hat{K}$.
For instance, at confidence level 95\%, the uncertainty under the
sup-norm metric exceeds 71.7\% of $\hat{K}/2$ on at least 50\% of the domain.

\begin{table}
\centering
	\caption{ \protect{\label{tab:integrals}}
	Confidence bounds for quantiles and the mean of the uncertainty of the minimax emulator 
	$\hat{f}_{\hat{K}}$ for CAM}
	\begin{tabular}{ll|rrr|r}
		& & \multicolumn{4}{c}{95\% lower confidence bound}  \\
		\cline{3-6} norm & units & lower quartile & median & upper quartile & average   \\
		\hline
		Euclidean & $\hat{K} / 2$ & 1.462 & 1.599 & 1.732 & 1.599 \\
		supremum & $\hat{K} / 2$ & 0.648 & 0.716 & 0.781 & 0.715 \\
		Euclidean & $\hat{\gamma}$ & 0.044 & 0.049 & 0.053 & 0.049 \\
		supremum & $\hat{\gamma}$ & 0.048 & 0.053 & 0.058 & 0.053
	\end{tabular}

\vspace{15px}
	\parbox[left]{30em}{
		\footnotesize
	   Column~1: distance metric $d$ used for the Lipschitz constant.
	   Columns~3--5: binomial lower 95\% confidence bounds for quartiles of the  uncertainty,
	   obtained by inverting binomial tests.
	   Column~6: 95\% lower 95\% confidence bound for the integral of the  uncertainty
	   over the entire domain $\unitcube$, based on inverting $z$-tests.
	   Columns 3--6 are expressed as a fraction of the quantity in column~2.
	   Results are based on 10,000 uniform random samples from $\unitcube$.
}

\end{table}

\subsection{Uncertainty relative to typical values}

We have focused on taking $\epsilon$ to be a fraction of $K$ or $\hat{K}$.
When $\epsilon$ is chosen that way, sections \ref{sec:Bounds-on-num-evals}
and \ref{sec:Bounds-on-error}  establish
conditions under which no emulator can be guaranteed to
replicate the variation of $f$.
Emulators are generally constructed to capture the complexity of the model: tracking
its variability.
That suggests approximating $f$ to within a fraction of its variation, which is why
we have calibrated $\epsilon$ to $\hat{K}$.
If the goal were to approximate $f$ to within a fraction of its mean,
and its mean is large compared to its variation,
approximating $f$ globally by its sample mean might suffice.
Then it might make sense to set
$\epsilon$ to be a fraction of a typical value of $f$, for instance, $\bar{\gamma}$ or
the sample mean
\[
	\hat{\gamma} = \frac{1}{\#X} \sum_{x \in X} f(x).
\]
The last 2 rows of Table~\ref{tab:integrals} list confidence bounds for percentiles
of the  uncertainty as a fraction of $\hat{\gamma}$.

Similarly, for $\epsilon$ chosen suitably, inequality~(\ref{eq:M_lower_bound})
gives a lower bound on $\M$ for approximating $f$ within a fraction of its
typical value, rather than within a fraction of its observed variation.
(Of course, the resulting bounds can be made arbitrarily small by
adding a sufficiently
large constant to $f$.
One reason we think it is more interesting to
calibrate $\epsilon$ as a fraction of $K$ or $\hat{K}$ is that
the results are invariant under affine transformations of $f$.)

For the CAM model, this lower bound on $\M$
is trivial when $\epsilon$ is a large fraction of the typical value of $f$, but grows
rapidly as the fraction decreases (table~\ref{tab:typical_f_burden}).

\begin{table}
\centering
	\caption{ \protect{\label{tab:typical_f_burden}}
	Minimum computational burden for the CAM model.}
	\begin{tabular}{llr}
		norm & $\epsilon/\hat{\gamma}$ & lower bound on $\M$ \\
		\hline
		Euclidean &	$0.02 $ & $3.6 \times 10^{12}$ \\
		&	$0.04$ & 1,720,354 \\
		&	$0.06$ & $345$ \\
		&	$0.08$ & $1$ \\
		\hline
		supremum &	$0.02$ & $8.6 \times 10^{10}$ \\
		&	$0.04$ & 413,595 \\
		&	$0.06$ & $83$ \\
		&	$0.08$ & $1$
	\end{tabular}
\end{table}

\subsection{Other uses for $e_{\kappa}^-$ and $e_{\kappa}^+$}
We have primarily used $e_{\kappa}^+$ and $e_{\kappa}^-$ to construct the minimax
emulator and find its uncertainty.
But if $f$ is no less regular than it was observed to be,
$e_{\hat K}^+$  is a pointwise upper bound on $f$
and $e_{\hat K}^-$ is a pointwise lower bound on $f$.
Moreover,  if $f$ is no less regular than the data require it to be,
$\max_{w \in \unitcube} e_{\hat K}^+(w)$
is a global upper bound on $f$ and
$\min_{w \in \unitcube} e_{\hat K}^-(w)$
is a global lower bound on $f$.

Maximizing $e_{\hat K}^+$ or minimizing  $e_{\hat K}^-$ exactly
may not be tractable.
For sup-norm, we can use the techniques from section \ref{sub:climate_app_1} to bound
these extrema from above and below:
for the CAM model, those upper and lower bounds on
$e_{\hat K}^+$ are equal, as they are for  $e_{\hat K}^-$.
The maximum of
$e_{\hat K}^+$ is $253.78$ and the minimum of $e_{\hat K}^-$ is $211.88$.


\section{Conclusions} \label{sec:conclusions}
We find a lower bound on the minimum (over emulators) maximum (over functions
that agree with the data and are as regular as the data allow) error of emulators of a function $f$
based on $n$~observations.
This ``mini-minimax'' uncertainty is optimistic because it assumes that $f$
has the smallest Lipschitz constant consistent with the data.
The mini-minimax uncertainty is an attainable bound on the error of the
best emulator of $f$ at $w$:  for any emulator $\hat{f}$, there is a function
$g$ that is at least as regular as $f$,
that agrees with $f$ at the $n$~observations, and for which
$|\hat{f}(w) - g(w)|$
is at least this mini-minimax value.

In some problems, \emph{every} emulator
based on any tractable number of observations of $f$
has large maximum uncertainty (and the  
uncertainty is large over much of the domain),
even if $f$ is as regular as the data allow.
That is, there are functions $g$ and $h$ that agree perfectly
with the observations, are as regular as the observations permit,
and yet differ by a large amount at some point in the domain of $f$.

We give sufficient conditions under which even the best possible emulator
has large uncertainty.
The conditions depend only on the observed values of $f$; they can be computed
from the same observations used to train an emulator, at a cost that typically is small
compared with the cost of generating those observations.
The conditions are sufficient but not necessary, because $f$ could be
less regular than any finite set of observations reveals it to be.
It is not possible to give necessary conditions that depend only on the observed
values of $f$; a priori bounds on the regularity of $f$ would be needed.

The conditions seem likely to hold for many high-consequence applications.
Indeed, we show quantitatively that the conditions hold for a large climate-modeling 
dataset.
When the maximum uncertainty in approximating $f$ everywhere
by a constant---the value of
$f$ at the center of the domain---is no larger than the maximum uncertainty
in approximating $f$ from any tractable number of observations,
emulators may not be useful.
No emulator can then
reliably model $f$ as a function of its input $w \in \unitcube$.

Common techniques for assessing the accuracy of emulators (e.g., posterior variance or
performance on hold-out data)
understate the true uncertainty,
because they make strong assumptions about $f$ that are
based neither on the observations nor on known
properties of $f$, or because they focus on average error rather than worst-case error.
However, as section~\ref{sec:extensions} shows, even the average uncertainty
and quartiles of the uncertainty for the CAM model are quite large.

The mini-minimax uncertainty is a one-sided tool: if this uncertainty is large,
the data do not constrain $f$ well, while if it is small, the data constrain
$f$ \emph{only} if it is no less regular than the data collected so far show it must be.
That said, if the mini-minimax uncertainty is uncomfortably large, there might be ways
to reduce it.
For instance, if the lower bound (\ref{eq:M_lower_bound}) on the computation burden
required to reduce the uncertainty to a useful level $\epsilon$
is affordable, one might collect more data.
Provided the new data do not increase $\hat{K}$ substantially, the mini-minimax uncertainty
can be reduced at will.
But when $p$ is large, the lower bound is likely to be large, because it grows exponentially
with $p$.
If observing $f$ requires a real-world experiment, new technology might be required to make
a useful number of additional observations affordable.
When observing $f$ involves running a simulator, collecting enough additional
data to reduce the uncertainty to a reassuring value might require
not only recruiting additional computational resources but also 
reducing the computational cost of each simulation---substantially.

In some cases, clever strategies can reduce the cost of computing $f$, at least to some known
degree of approximation,
but that is not always so.
Cost reductions of orders of magnitude might require 
reducing the complexity of $f$. 
Reducing the dimension $p$ of the domain of $f$ is especially helpful, because reducing $p$ 
pays exponential
dividends.
But it requires scientific justification: In general, eliminating parameters from a model
entails bias in the model with no \emph{a priori} limit.
It is hard to calibrate the tradeoff between fitting a model that is constrained by the data but
is known or suspected to be overly simplistic---and 
therefore biased---and a model that has lower bias but cannot be estimated reliably from an 
affordable number of data.
Subject-matter knowledge is key.


Without increasing the number of observations or revising the model,
reducing the uncertainty of emulators
requires either more information about $f$\footnote{%
	Common additional conditions include the following: parameters have
	only low-order interactions; the second derivative has an upper bound;
	the third derivative has a limited number of knots; the integral of
	the squared derivative of the model is bounded~\cite{Lamboni2012}.
	There are problems in which conditions like these may reflect
	actual knowledge about $f$.
	However, such conditions tend to be difficult to verify:
	simulation is perhaps most valuable when the underlying equations
	are not amenable to mathematical analysis.
}
or changing the measure of uncertainty---changing the scientific
question.
Finally, approximating $f$ pointwise is not usually the ultimate scientific goal.
More important questions about $f$ might be answered more
directly.\footnote{%
	For example, for global optimization---finding maxima or
	minima---a form of adaptive sampling known as multi-start methods
	yields good results~\cite{Hickernell1997}.
}
These tactics are application-specific: the
underlying science dictates the conditions that actually hold
for $f$ and the questions about $f$ that matter.
\vspace{28em}
\pagebreak{}

\appendix
\section{\label{sec:proofs}Proofs}

\noindent
For real $\chi$ and $\rho$, define the interval
\[
	I(\chi,\rho)\equiv
	\begin{cases}
		\left[\chi-\rho,\chi+\rho\right], & \rho\ge0\\
		\emptyset, & \text{otherwise.}
	\end{cases}
\]
If $I$ is an interval, $\mu(I)$ denotes its length; for instance,
$\mu(I(\chi, \rho)) = \max(0, 2\rho)$.

\begin{restatable}{lem}{intervallem}
\label{lem:interval}
\textup{Fix $\alpha \in[0,1]$, $\rho_{1},\ldots,\rho_{n}\in[0,\infty)$
and $\chi_{1},\ldots,\chi_{n}\in\mathbb{R}$. Let $I_{1}\equiv\bigcap_{i=1}^{n}I(\chi_{i},\rho_{i})$
and $I_{ \alpha }\equiv\bigcap_{i=1}^{n}I(\chi_{i}, \alpha  \rho_{i})$. Then
$ \alpha \mu\left(I_{1}\right)\ge\mu\left(I_{ \alpha }\right)$.}
\end{restatable}

\begin{proof}
Because the intersection of intervals is itself an interval, there
exist $\chi_0$ and $\rho_0$ satisfying
\[
I_{ \alpha }=I(\chi_0, \rho_0).
\]
Fix $i\in{1,\ldots,n}$. Then
\[
I(\chi_0, \rho_0)\subset I(\chi_{i}, \alpha  \rho_{i}).
\]
It follows that
\[
\chi_0 - \rho_0 \ge \chi_i - \alpha  \rho_i.
\]
Then
\[
 \alpha \left(\rho_{i}-\frac{\rho_0}{ \alpha }\right)\ge \chi_{i}-\chi_0.
\]
Because $ \alpha \le 1$ and $\rho_{i} \ge 0$,
\[
\rho_{i}-\frac{\rho_0}{ \alpha }\ge \chi_{i}-\chi_0.
\]
Finally,
\[
\chi_0-\frac{\rho_0}{ \alpha }\ge \chi_{i}-\rho_{i}.
\]
By symmetric reasoning we also have
\[
\chi_0 + \frac{\rho_0}{ \alpha }\le \chi_{i}+\rho_{i}.
\]
Therefore,
\[
I\left(\chi_0, \frac{\rho_0}{ \alpha }\right)\subset I(\chi_{i},\rho_{i}).
\]
Because $i$ was arbitrary,
\[
I\left(\chi_0,\frac{\rho_0}{ \alpha }\right)\subset I_{1}.
\]
Hence,
\[
\mu \left ( I_1 \right ) \ge \mu\left(I\left(\chi_0, \frac{\rho_0}{ \alpha }\right)\right)=\frac{2\rho_0}{ \alpha }
   = \frac{\mu \left ( I_\alpha \right )}{\alpha}.
\]
\end{proof}
Lemma~\ref{lem:interval} is used in the proof of Theorem~\ref{thm:main}, below.
\vspace{20px}

\minimaxEquivProp*
\begin{proof}
${}$\\[1ex]
\noindent
{\emph{Step 1: $e_{\kappa}^{+}$ and $e_{\kappa}^{-}$ are Lipschitz continuous with
               constant $\kappa$.}} \\
	For $v, w \in \unitcube$, $\exists x, y \in X$ satisfying
	\[
	     e_{\kappa}^{+}(v)=f(x)+ \kappa d(x,v) \mbox{ and } e_{\kappa}^{+}(w)=f(y)+
	     \kappa d(y,w).
	\]
	Suppose without loss of generality that $e_{\kappa}^+(v) \ge e_{\kappa}^+(w)$.
	By construction, $e_{\kappa}^{+}(v) \le f(y)+ \kappa d(y,v)$.
	Hence
	\begin{eqnarray}
	   0 \le e_{\kappa}^+(v) - e_{\kappa}^+(w)
	         &\le & f(y) + \kappa d(y, v) - e_{\kappa}^+(w) \nonumber \\
	         &= & f(y) + \kappa d(y, v) - f(y) - \kappa d(y,w) \nonumber \\
	         &\le & \kappa (d(y, v)-d(y,w)) \nonumber \\
	         &\le & \kappa d(v, y), \nonumber
	\end{eqnarray}
	by the triangle inequality.
	Hence $e_{\kappa}^+$ has Lipschitz constant $\kappa$.
	An analogous argument shows that $e_{\kappa}^-$ also has Lipschitz constant $\kappa$.

\vspace{1ex}
\noindent
{\emph{Step 2: $e_{\kappa}^{+}$ and $e_{\kappa}^{-}$
       agree with $f$ on $X$.}} 
       (Hence, $\mmfk = (e_{\kappa}^{+} + e_{\kappa}^{-})/2$ agrees with $f$ on $X$.)\\
	We have
	$$\kappa \ge \hat{K} \equiv \max_{x, y \in X: x \ne y} \frac{|f(x) - f(y)|}{d(x, y)},$$
	and hence $|f(x) - f(y)| \le \kappa d(x,y)$ for all $x, y, \in X$.
	Thus
	$$\min_{x \in X} [f(x) + \kappa d(x, y)] = \min \left \{ f(y), \min_{x \in X, x \ne y} [f(x) + \kappa d(x, y)] \right \}
	    = f(y).
	$$
	Similarly, $\max_{x \in X} [f(x) - \kappa d(x, y)] = f(y)$ for $y \in X$.
	Hence, $e_{\kappa}^{+}(y) = e_{\kappa}^{-}(y) = f(y)$ for $y \in X$.
	Since, as shown in step~1, $e_{\kappa}^{+}$ and $e_{\kappa}^{-}$ are
	Lipschitz with constant $\kappa$,
	$e_{\kappa}^{+}$ and $e_{\kappa}^{-} \in \F_{\kappa}$.


\vspace{1ex}
\noindent
{\emph{Step 3: $e_{\kappa}^-$ is the pointwise infimum of $\F_{\kappa}$ and
    $e_{\kappa}^+$ is the pointwise supremum of $\F_{\kappa}$.}} \\
	Suppose to the contrary that there exists
	$w\in\unitcube$, $x \in X$, and $g\in\F_{\kappa}$ for which
	\[
         	g(w)>f(x)+\kappa d(x,w).
	\]
	Recall that  $g \in \F_{\kappa}$ implies that $g(x) = f(x)$ $\forall x \in X$.
	Hence
        \[
              g(w)-g(x) >  f(x) + \kappa d(x,w) - f(x)  = \kappa d(x, w).
        \]
        That is, $g$ has a Lipschitz constant greater than $\kappa$, a contradiction.
	Hence, $e_{\kappa}^+(w) = \sup\{g(w): g \in \F_{\kappa}\}$ for all $w \in \unitcube$.
	The same argument, \emph{mutatis mutandi}, shows that
	$$e_{\kappa}^-(w) = \inf\{g(w): g \in \F_{\kappa}\} \mbox{ for all } w \in \unitcube.$$

\noindent
{\emph{Step 4: The maximum uncertainty of $\mmfk$ at $w$, $\E_\kappa (w; \mmfk)$, equals
                        $[e_\kappa^+(w) - e_\kappa^-(w)]/2$.
}}
\begin{align}
	\E_\kappa (w; \mmfk) & \equiv \sup_{g \in \F_\kappa(w)} \left| \mmfk(w) - g(w) \right| \nonumber \\
	& = \max \left \{ \sup_{g \in \F_\kappa(w)} g(w) - \mmfk(w),  \mmfk(w) -
	      \inf_{g \in \F_\kappa(w)} g(w) \right \}
	   \nonumber \\
	& =  \max \left \{ e_{\kappa}^+(w) - \mmfk(w), \mmfk(w) - e_{\kappa}^-(w) \right \}  \label{eq:sup_to_e} \\
	& =  \max \left \{ e_{\kappa}^+(w) - \frac{e_{\kappa}^+(w)+e_{\kappa}^-(w)}{2}, \frac{e_{\kappa}^+(w)+e_{\kappa}^-(w)}{2} - e_{\kappa}^-(w) \right \} \nonumber \\
	& =  \frac{e_{\kappa}^+(w) - e_{\kappa}^-(w)}{2}. \nonumber
\end{align}
Equality (\ref{eq:sup_to_e}) follows from step~3.

\vspace{1ex}
\noindent
{\emph{Step 5: The minimax uncertainty at $w$, $\E_\kappa (w)$, equals
                        $[e_\kappa^+(w) - e_\kappa^-(w)]/2$.}}

\noindent
Suppose $\hat{f}(w) > \mmfk(w)$.
Then $$|\hat{f}(w) - e_{\kappa}^-(w)| > \frac{e_\kappa^+(w) - e_\kappa^-(w)}{2} = \E_\kappa (w; \mmfk).$$
Suppose $\hat{f}(w) < \mmfk(w)$.
Then $$|\hat{f}(w) - e_{\kappa}^+(w)| > \frac{e_\kappa^+(w) - e_\kappa^-(w)}{2} = \E_\kappa (w; \mmfk). $$
Hence, $\mmfk(w)$ is minimax, and $\E_\kappa (w) = \E_\kappa(w, \mmfk) = [e_\kappa^+(w) - e_\kappa^-(w)]/2$.

\end{proof}

\begin{restatable}{lem}{uncertainty_ball_lem}
\label{lem:uncertainty_ball}
For $\kappa \ge 0$,
  $$\E_{\kappa}(w) = \frac{1}{2}\mu\left(\bigcap_{x\in X}I\left(f(x),\kappa d(x,w)\right)\right).$$
\end{restatable}
\begin{proof}
\begin{eqnarray}
	\E_{\kappa}(w) & = & \frac{1}{2}\left[e_{\kappa}^{+}(w)-e_{\kappa}^{-}(w)\right]
	                               \nonumber \\
	 & = & \frac{1}{2}\left\{ \min_{x\in X}\left\{ f(x)+\kappa d(x,w)\right\} -
	           \max_{x\in X}\left\{ f(x)-\kappa d(x,w)\right\} \right\} \nonumber \\
	 & = & \frac{1}{2}\mu\left(\left[\max_{x\in X}\left\{ f(x)-\kappa d(x,w)\right\} ,
	      \min_{x\in X}\left\{ f(x)+\kappa d(x,w)\right\} \right]\right)\nonumber \\
	 & = & \frac{1}{2}\mu\left(\bigcap_{x\in X}I\left(f(x),\kappa d(x,w)\right)\right).\nonumber
\end{eqnarray}
The first equality follows from proposition~\ref{prop:supequiv}.
\end{proof}

\thmmain*
\begin{proof}
Let $w^{\star}\equiv\arg\max_{w}\E_{\hat{K}}(w)$.
Then
\begin{eqnarray}
	 \E_K( \hat{f})
	         & \ge &  \E_K( \mmfK)\nonumber \\
		 & = & \E_{K}(w^{\star}) \nonumber \\
		 & \ge & \frac{K}{\hat{K}}\cdot \E_{\hat{K}}(w^{\star})\label{eq:thm_formerly_2}\\
		 & \ge & \frac{K}{\hat{K}}\cdot  \lambda \hat{K}\label{eq:thm_formerly_3}\\
		 & = &  \lambda  K.\nonumber
\end{eqnarray}
Inequality (\ref{eq:thm_formerly_3}) follows from (\ref{eq:thm_formerly_2}) by hypothesis.
Inequality (\ref{eq:thm_formerly_2}) is a consequence of lemma~\ref{lem:interval}:
Let $\alpha=\hat{K}/K \le 1$.
For, $i=1,\ldots,\#X$, let $\rho_{i}=f(x_{i})$ and $\chi_{i}=Kd(x,w)$.
Then, by lemma~\ref{lem:uncertainty_ball}, $\mu(I_{1})/2=\E_{K}$
and $\mu(I_{\alpha})/2=\E_{\hat{K}}$.
\end{proof}

\climateubprop*
\begin{proof}
	Fix $w\in\unitcube$.
	Let $w_{(i)}$ denote the $i^{\text{th}}$ component of $w$.
	Then
	\begin{align*}
		d(v,w) & =\max_{i\in\left\{ 1,\ldots,p\right\} }\left | v_{(i)}-w_{(i)} \right | \\
		 & \le\max_{i\in\left\{ 1,\ldots,p\right\} }\max_{\delta\in\{0,1\}}\left|v_{(i)}-\delta\right|\\
		 & =\max_{i\in\left\{ 1,\ldots,p\right\} }\max_{y\in\{ \cornerA , \cornerB \}}
		            \left| v_{(i)}-y_{(i)}\right|\\
		 & =\max_{y\in\{ \cornerA , \cornerB \}}\max_{i\in\left\{ 1,\ldots,p\right\} }
		             \left|v_{(i)}-y_{(i)}\right|\\
		 & =\max_{y\in\{ \cornerA , \cornerB \}}d(v,y)\\
		 & =\max(d(v, \cornerA ),d(v, \cornerB )).
	\end{align*}
	Hence,
	\begin{align}
	       \E_{\hat{K}}(w)
	            =&
	      \frac{1}{2}\mu\left(\bigcap_{x\in X}I\left(f(x),\hat{K}d(x,w)\right)\right)
	               \label{eq:prop_by_lemma6}\\
	      \le & \frac{1}{2}\mu\left(\bigcap_{x\in X}I\left(f(x),\hat{K}\tilde{d}(x)\right)\right)
	               \label{eq:prop_by_choice_of_d}\\
	       = & \frac{1}{2}
	             \left\{ \min_{x\in X}
	                   \left[f(x)+\hat{K}\tilde{d}(x)\right]-
	                       \max_{x\in X}\left[f(x)-\hat{K}\tilde{d}(x)
	                    \right]
	             \right\}  \nonumber
	\end{align}
	where (\ref{eq:prop_by_lemma6}) follows from lemma~\ref{lem:uncertainty_ball}.
	Because the right-hand side of this inequality does not depend on
	$w$, the proposition follows by taking suprema.
\end{proof}

\section{\label{sec:borehole_K} The Lipschitz constant $K$ for the Borehole function}
The Borehole function is
\[
	f_0(H_u, H_\ell, T_u, T_\ell, r, r_w, L, K_w)
		= \frac{2\pi T_{u}\left(H_{u}-H_{l}\right)}
		{\log\left(r/r_{w} \right) \left(1+\frac{2LT_{u}}{\log\left(r/r_{w}\right)r_{w}^{2}K_{w}}+\frac{T_{u}}{T_{\ell}}\right)}.
\]
The variables are restricted to the ranges in table~\ref{tab:borehole_domain}.
We rescale $f_0$ so that its inputs range over the 8-dimensional unit hypercube $[0,1]^8$;
the resulting function is denoted $f$.

In $\ell_\infty$, because $f$ is differentiable and $\dom(f)$ is convex,
\begin{align*}
	K &= \sup_{w \in \dom(f)} \|Df(w)\|_\infty
      = \sup_{w \in \dom(f)} \sum_{i=1}^{8}\left|\frac{\partial}{\partial w_{i}}f(w)\right|.
\end{align*}
Let
\[
	H = 2\pi(H_u - H_\ell),
\]
\[
	R = \log(r / r_w),
\]
\[
	M = 2 L / K_w,
\]
\[
	t = T_\ell^{-1} + T_u^{-1}
\]
and
\[
	S = M + Rr_{w}^{2}t.
\]
Now
\begin{align*}
	f_0(H_u, H_\ell, T_u, T_\ell, r, r_w, L, K_w) & =\frac{Hr_{w}^{2}}{S}.
\end{align*}
We bound each partial derivative of $f$ using the ranges of the input variables:
\begin{align*}
	\left|\frac{\partial f_0}{\partial H_{\ell}}\right|
	= \left|\frac{\partial f_0}{\partial H_{u}}\right|
	=\frac{2\pi r_{w}^{2}}{S}
	\le 0.76 & \,\,\Longrightarrow \,\, 
	\left|\frac{\partial f}{\partial H_{\ell}}\right| =
	\left|\frac{\partial f}{\partial H_{u}}\right| \le 91.2 \\
\left|\frac{\partial f_0}{\partial T_{u}}\right|
	=\frac{HRr_{w}^{4}}{S^{2}T_u^2}
	\le 0.01 & \,\,\Longrightarrow \,\, \left|\frac{\partial f}{\partial T_{u}}\right| \le 0.01 \\
	\left|\frac{\partial f_0}{\partial T_{l}}\right|
	=\frac{HRr_{w}^{4}}{S^{2}T_\ell^2}
	\le 0.13 & \,\,\Longrightarrow \,\, \left|\frac{\partial f}{\partial T_{l}}\right| \le 6.8 \\
\left|\frac{\partial f_0}{\partial r}\right|
	= \frac{Hr_{w}^{4}t}{S^{2}r}
	\le 0.01 & \,\,\Longrightarrow \,\, \left|\frac{\partial f}{\partial r}\right| \le 290.8 \\
\left|\frac{\partial f_0}{\partial r_w}\right|
	=\frac{Hr_w^3t}{S^{2}} + \frac{2H}{S(1/r_w+Rr_wt/M)}
	\le 4050.2 & \,\,\Longrightarrow \,\, \left|\frac{\partial f}{\partial r_w}\right| \le 405.0 \\
\left|\frac{\partial f_0}{\partial L}\right|
	=\frac{2Hr_{w}^{2}}{S^{2}K_w}
	\le 0.34 & \,\,\Longrightarrow \,\, \left|\frac{\partial f}{\partial L}\right| \le 190.4 \\
\left|\frac{\partial f_0}{\partial K_w}\right|
	=\frac{2LHr_{w}^{2}}{S^{2}K_w^2}
	\le 0.06 & \,\,\Longrightarrow \,\, \left|\frac{\partial f}{\partial K_w}\right| \le 123.7.
\end{align*}
Summing these upper bounds for the partial derivatives of $f$ yields
\[
	\sup_{w\in\dom(f)} \|Df(w)\|_\infty < 1200.
\]
Moreover, for $w_0 = (1100, 700, 115547, 116, 100, 0.15, 1120, 12045)$,
\[
	\|Df(w_0)\|_\infty = 944.
\]
Hence, for the rescaled borehole function $f$,
\[
	944 \le K \le 1200.
\]

\vspace{10px}
\subsection*{Acknowledgments}
We thank the Associate Editor and both Referees for their helpful comments about this paper.

\newpage
\bibliographystyle{unsrt}
\bibliography{references}
\vfill

\end{document}